\documentclass[letterpaper,10pt,conference]{ieeeconf}
\IEEEoverridecommandlockouts      
\author{Shuang~Gao and Peter E. Caines 
\thanks{*This work is supported by NSERC (Canada) Grant RGPIN-2024-06612.}
\thanks{Shuang Gao is with the Department of Electrical Engineering, Polytechnique Montreal, GERAD, and UNIQUE (Unifying Neuroscience and Artificial Intelligence - Quebec), Montreal, QC, Canada.   {Email: \tt\small    shuang.gao@polymtl.ca}. Peter E. Caines is with the Department of Electrical and Computer Engineering, McGill University \& GERAD,
   Montreal, QC, Canada. 
         {Email: \tt\small   peterc@cim.mcgill.ca}.}
\thanks{SG gratefully acknowledges Roland P. Malham\'e and Evelyn Hubbard for their helpful feedback on this work. }
}%

\usepackage[noadjust]{cite}

\usepackage{graphics}
\usepackage{epstopdf}

\usepackage{amsmath}
  \usepackage[font=footnotesize]{subfig}

\usepackage{hyperref}
\usepackage{enumerate}
\usepackage{color}

\newcommand{\BR}{ \operatorname{R}}  

\newcommand*\TRANS{{\mathpalette\doTRANS\empty}}
\makeatletter
\newcommand*\doTRANS[2]{\raisebox{\depth}{$\m@th#1\intercal$}}
\makeatother
\newcommand\MATRIX[1]{\begin{bmatrix} #1 \end{bmatrix}}

\usepackage{url}
\usepackage{amssymb,mathtools}

\usepackage[standard,thmmarks,amsmath]{ntheorem}
\usepackage{times}
\usepackage{dsfont}
\usepackage{tikz}

\begin{document}
%
\title{Transmission Neural Networks: Inhibitory and Excitatory Connections}
%
%
%
\maketitle 
\begin{abstract}
This paper extends the Transmission Neural Network model proposed by Gao and Caines in  \cite{ShuangPeterTransNN22, ShuangPeterTransNNControl25,ShuangPeterTransNNCDC25} to incorporate inhibitory connections and neurotransmitter populations. 
The extended network model contains binary neuronal states,  transmission dynamics, and inhibitory and excitatory connections. Under technical assumptions, we establish the characterization of the firing probabilities of neurons, and show that such a characterization in the presence of inhibition can be equivalently represented by a neural network where each neuron has a continuous state of dimension 2. In addition, we incorporate neurotransmitter populations into the modeling  and  establish the limit network model when the number of neurotransmitters at all synaptic connections goes to infinity. Finally, sufficient conditions for stability and contraction properties of the limit network model are established. 
\end{abstract}

\section{Introduction}

Modelling neuronal systems is important to understand intelligence  and to analyze and control such systems. Networks of neurons can learn input-output relations in the context of artificial neural networks \cite{block1962perceptron,rumelhart1986learning,lecun2015deep,goodfellow2016deep}. Moreover, recent advances show that combining detailed brain networks,  such as the  Drosophila connectome,  with relatively simple neuronal dynamics can predict  neural activities associated with specific  sensorimotor processing~\cite{shiu2024drosophila}. 
 Neuronal models with different levels of abstraction have been proposed to characterize the behaviors of neuronal systems \cite{gerstner2014neuronal,Bullo2025}, ranging from the detailed descriptions of the dynamics  of individual neurons by Hodgkin and Huxley \cite{hodgkin1952quantitative} to network-level models characterizing interactions \cite{gerstner2014neuronal,hopfield1982neural,hopfield1984neurons}. 

Neural network models that adopt a binary-state representation of neuronal systems 
offer certain advantages: (a) one can focus on  network level properties as in the work of Hopfield \cite{hopfield1982neural} and those on Boltzmann Machines \cite{hinton1983optimal, ackley1985learning}
and (b)  continuous-valued neural networks can be binarized to provide more efficient algorithms and learning models \cite{courbariaux2015binaryconnect,hubara2016binarized}.
 In addition, the binary state for each neuron  can be naturally linked with a continuous value by taking the probability of neurons being activated as a neuronal state \cite{hopfield1984neurons,ShuangPeterTransNN22}, for which  control-theoretic properties including stability can be established (see e.g. \cite{Bullo2025,ShuangPeterTransNN22}). 
%

The Transmission Neural Network (TransNN) model proposed in \cite{ShuangPeterTransNN22,ShuangPeterTransNNControl25, ShuangPeterTransNNCDC25} has established a natural connection between   neural networks and virus spread models, where the connection of the nodes resembles the process of synaptic transmission. The work \cite{ShuangPeterTransNNControl25, ShuangPeterTransNNCDC25}  investigated how TransNN approximates stochastic neural networks with binary nodal states, and proposed TransNN-based approximate control algorithms for controlling these stochastic networks. 

Inhibition that suppresses the activity of neurons is essential for  neuronal systems \cite{eccles1954cholinergic,hartline1956inhibition, kandel2000principles}.
Inhibitory properties have been considered in  models of neurons with various formulations (e.g. the Wilson-Cowan model for neuronal populations \cite{wilson1972excitatory,wilson1973mathematical} and models with continuous-time stochastic dynamics   \cite{cottrell1992mathematical,turova1997stochastic}).

\textcolor{black}{The main focus of the present paper is to generalize TransNN models  \cite{ShuangPeterTransNNControl25, ShuangPeterTransNNCDC25} to account for inhibitory connections. The central challenges are to identify a biologically motivated and mathematically tractable inhibition mechanism and to derive  approximate models that capture its effect.}

\textbf{Contribution:} This work extends the TransNN models with binary nodal states in \cite{ShuangPeterTransNNControl25,ShuangPeterTransNNCDC25} to include inhibition, and  identifies the corresponding  model for the probability of excitation. A key consequence of introducing inhibition is that the scalar-state simplification  in \cite{ShuangPeterTransNNControl25,ShuangPeterTransNNCDC25}  no longer applies directly.  We show that such networks with inhibition can be equivalently represented by neural networks where a two-dimensional nodal state is associated with each neuron and the Tuneable Log-Sigmoid activation function in \cite{ShuangPeterTransNN22}  with each synaptic connection.  Moreover, we incorporate neurotransmitter populations into an extended TransNN model and establish its limit model by letting the number of neurotransmitters at all synaptic connections go to infinity. Finally, sufficient conditions for stability and contraction properties of the limit network model are established.

\textbf{Notation}: $\BR$ denotes the set of real numbers. Let $\bar{\BR} \triangleq \BR \cup \{+\infty \}$ and $[n]\triangleq \{1,2,...,n\}$.   We use 
$W\triangleq [W_{ij}] \in \BR^{n\times n}$ to denote the matrix with its $ij^{\textup{th}}$ element specified by $W_{ij}$ for all $i, j \in [n]$. For a vector $v \in \bar{\BR}^n$, we use both $v_i$ and $[v]_i$ to denote its $i^{\text{th}}$ element. For a matrix $W \in \BR^{n\times n}$, $\|W\|_1 \triangleq \max_{j \in [n]} \sum_{i=1}^n|W_{ij}|$ and $\|W\|_\infty \triangleq \max_{i \in [n]} \sum_{j=1}^n|W_{ij}|$. For a vector $v\in \BR^n$, $\|v\|_1 \triangleq \sum_{i=1}^n|v_i|$ and $\|v\|_\infty \triangleq  \max_{i \in [n]}|v_i|$. 

\section{Transmission Dynamics with Inhibitory and Excitatory Connections}

Consider a network of $n$ neurons with interconnections through chemical synapses \cite{kandel2021principles}. The synaptic connection structure at time or layer{\footnote{We note that, throughout the paper,  time step $k$ can also be interpreted as layer $k$ in the context of  neural networks with multiple layers.}} $k\geq 0$ is represented by a directed   graph $\mathcal{G}^k = ([n], \mathcal{E}^k)$ with the node set $[n]\triangleq \{1,2,...,n\}$ and the edge set $\mathcal{E}^k\subset [n]\times [n]$, which may include self-loops. A directed connection from neuron $j$ to neuron $i$, denoted by the node pair $(i,j) \subset \mathcal{E}^k$, exists if the axon terminal of neuron $j$ has at least one synapse onto neuron $i$ at step $k$. Such synaptic connections could be either excitatory or inhibitory \cite{kandel2021principles}. Let $\mathcal{G}_h^k = ([n], \mathcal{E}_h^k)$ denote the subgraph of $\mathcal{G}^k$ with all inhibitory connections, and  $\mathcal{G}_e^k = ([n], \mathcal{E}_e^k)$  the subgraph of $\mathcal{G}^k$ with all excitatory connections. Furthermore, $\mathcal{E}_h^k \cup \mathcal{E}_e^k = \mathcal{E}^k$ and $\mathcal{E}_h^k \cap \mathcal{E}_e^k = \varnothing$, that is, a connection between two neurons will be either inhibitory or excitatory.  We allow self-loops to model autapses (i.e. synapses from a neuron onto itself)  in  neuronal systems \cite{bekkers2003synaptic}.  
Furthermore, since an autapse is either inhibitory or excitatory \cite{bekkers2003synaptic}, the self-loops considered are either  inhibitory or  excitatory, that is, for each node $i\in[n]$, its self-loop (i.e. the edge pair $(i,i)$) may appear in either $\mathcal{E}_h^k $  or $\mathcal{E}_e^k $, but not in both. 
Let $B_E^k $ denote the  binary-valued adjacency matrix  of  the subgraph $\mathcal{G}_e^k = ([n], \mathcal{E}_e^k)$ with all excitatory connections at step $k$, whose $ij^{\text{th}}$ element  is 1 if $(i, j) \in \mathcal{E}_e^k$, and $0$ otherwise.
Similarly, let $B_I^k$ denote the binary-valued adjacency matrix of the subgraph  $\mathcal{G}_h^k = ([n], \mathcal{E}_h^k)$ with all inhibitory connections at step $k$.
\subsection{Transmission Dynamics}
The state of a neuron $i \in [n]$ at step $k$ is denoted by a binary variable $X_i(k)$ that takes $1$ if the neuron fires and $0$ otherwise.
We introduce the transmission dynamics with both excitatory and inhibitory connections as follows:
\begin{equation} \label{eq:dyn-ex-ih-update}
\begin{aligned}
	  X_i(k+1) = ~& \Big(1-\prod_{j\in E_i^{ \circ k}}(1- W_{ij}^k X_j(k))\Big)  \\
	& \times  \prod_{j \in I_i^{ \circ k}}(1- W_{ij}^k X_j(k)) 
\end{aligned}
\end{equation}
where $E_i^{ \circ k}  \triangleq  \{j: (i,j)\in \mathcal{E}_e^k \} $ denotes the set of incoming neighbouring nodes of $i$ with excitatory connection that potentially includes node $i$ at step $k$, $I_i^{ \circ k}  \triangleq  \{j: (i,j)\in \mathcal{E}_h^k \} = \{j: (i,j)\in  \mathcal{E}^k - \mathcal{E}_e^k  \}$ denotes the set of incoming neighbouring nodes of $i$ with inhibitory connections that potentially includes node $i$ at step $k$,  and $W_{ij}^k \in\{0,1\}$ is the binary variable that represents the successful transmission  when taking value $1$, otherwise $0$.

\begin{remark}
The network model  with binary states in \eqref{eq:dyn-ex-ih-update} generalizes that in \cite{ShuangPeterTransNNControl25,ShuangPeterTransNNCDC25} to include inhibitory neuronal connections. 
In the  dynamics \eqref{eq:dyn-ex-ih-update}, a single effective inhibitory connection to a neuron  suppresses its excitation, whereas in the absence of inhibition, a single effective excitatory connection is sufficient to activate it.  \textcolor{black}{This rule provides a coarse-grained representation of inhibitory suppression and is qualitatively reminiscent of the functional effect of shunting inhibition (also known as the short-circuiting effect)  \cite[p.~290]{kandel2021principles}, in the sense that inhibitory input can suppress the effect of excitatory input. }   
\end{remark}

\begin{remark}[Functional Completeness] \label{rem:func-completeness}
The existence of a constant input is possible for neuronal systems (e.g. {persistent firing for working memory \cite{fuster1971neuron}}).
With a constant input  $1$ as well as excitatory and inhibitory connections, the interaction rules in \eqref{eq:dyn-ex-ih-update} can form an NOR gate with a simple network structure as illustrated in Fig.~\ref{fig:NOR}. 
\begin{figure}[htb]
\centering
	\includegraphics[width=8cm]{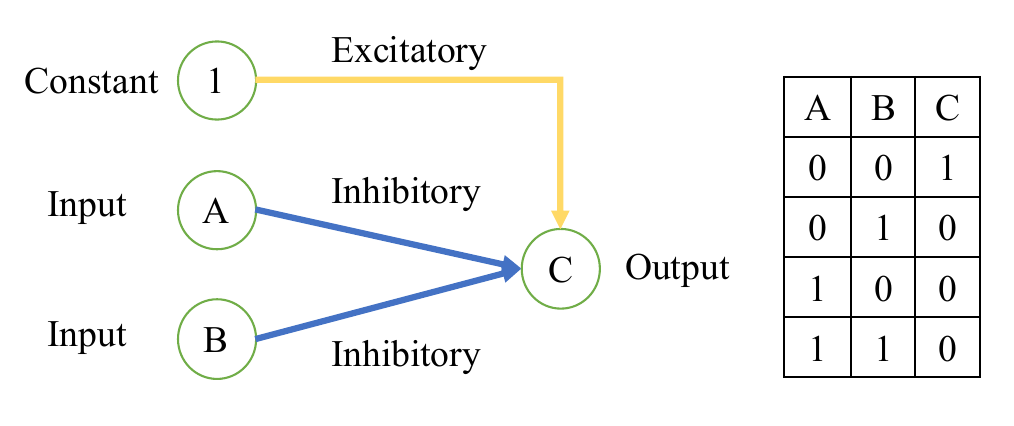}
	\caption{NOR gate  with inputs $A$ and $B$, and output $C$,  created with inhibitory connections and a constant signal 1. } \label{fig:NOR}
\end{figure}
Since any Boolean function can be implemented by combining NOR gates (i.e. the NOR gate is functionally complete), the interaction rule in \eqref{eq:dyn-ex-ih-update} with the constant signal~$1$ can reproduce any Boolean function by appropriate network designs. 
\end{remark}

\subsection{Stochasticity in the Transmission Dynamics}
Now we consider the case where  the state $X_i(k)$ and the transmission $W_{ij}^k$  for $i \in [n]$ and for $k\geq 0$ are stochastic. Let $X(k) \triangleq [X_1(k), \cdots, X_n(k)]^\TRANS$. 

 We introduce the assumptions on the independence of the transmissions and that of the states.  
\begin{description}
\item[\bf(A1)] (Memoryless Transmission) For any $k> 0$,	 $W^k\triangleq [W_{ij}^k]$ is independent of $\{W^t: 0\leq t < k\}$ and $\{X{(t)}: 0\leq t < k \}$. $W^0$ is independent of $X(0)$. 
\item[\bf (A2)] (Transmission Conditional Independence) For any $k\geq 0$, the  binary random variables $\{W_{ij}^k:i,j \in [n]\}$ representing the transmissions are jointly conditionally  independent given the current states $X(k)$. 
\end{description}

\begin{remark}
The assumption (A1) introduces the independence of the  transmission at the current step (or layer) from all the past   transmissions and  states.   (A2) imposes the conditional independence of transmissions across different links given the current state.    
{A special case is as follows: $W_{ij}^k$  is determined by flipping a biased coin independently with probability $w_{ij}^k$ being  head corresponding to the value $1$,  and  $X_i(k)$ is determined  similarly. %
}
\end{remark}
Under  the assumption (A1), the transmission dynamics in \eqref{eq:dyn-ex-ih-update} are {Markovian} and  
\begin{equation}
\begin{aligned}
	 \mathbb{E}(X_i(k+1)| (X(t))_{t\in [k]}) 
	 =  \mathbb{E}(X_i(k+1)| X(k)).
\end{aligned}
\end{equation} 
Furthermore, under  (A2), 
we have  
\begin{equation*}
\begin{aligned}
\mathbb{E}(X_i(k+1)|X(k)) 
	& =\Big(1-\prod_{j\in E_i^{ \circ k}}\mathbb{E}(1- W_{ij}^k X_j(k)| X(k)) \Big)  \\
	&\qquad \times  \prod_{j \in I_i^{ \circ k}}\mathbb{E}(1- W_{ij}^k X_j(k) | X(k) ) .
	\end{aligned}
\end{equation*} 
For a state configuration $q\in\{0,1\}^n$, the conditional probability of reaching $q$ is given by
\begin{equation}\label{eq:conditional-config}
\begin{aligned}
	\text{Pr}(&X(k+1)=q|X(k))   = \prod_{i=1}^n \text{Pr}(X_i(k+1)=q_i|X(k)) 
\end{aligned}
\end{equation}
where the equality is due to the conditional independence of the transmissions $\{ W_{ij}^k: i,j \in [n]\}$ assumed in (A2). %

\begin{proposition}\label{prop:cond-prob-infection}
Assume (A1) and (A2) hold. 
Given a state configuration  ${x} \in \{0,1\}^n$ at step $k$, the transition probability to a state configuration $q\in\{0,1\}^n$ is  given by 
\begin{equation}\label{eq:MarkovTranProb2}
\begin{aligned}
	\textup{Pr}(X(&k+1)=q|X(k)=x)\\ 
	& = \prod_{i=1}^n \Big(q_i\rho_i(k+1) + (1-q_i) (1-\rho_i(k+1)) \Big)
\end{aligned}
\end{equation}
where
 \begin{equation}\label{eq:prob-configuration}
 \begin{aligned}
\rho_i(k+1) & = \Big(1-	 \prod_{j\in E_i^{ \circ k}} (1- w_{ij}^k {x}_j)\Big)\\
 	& \qquad \times \prod_{j\in I_i^{ \circ k}} (1- w_{ij}^k {x}_j), \quad i\in [n]
 \end{aligned}
\end{equation}
and   
$
w_{ij}^k\triangleq \textup{Pr}(W_{ij}^k=1 {| X_j(k)=1})
$.
\end{proposition}
\begin{proof}
	Following \eqref{eq:conditional-config}, we note that
\begin{equation} 
	\begin{aligned}
	\textup{Pr}(X(&k+1)=q|X(k)=x)\\ & = \prod_{i=1}^n \textup{Pr}(X_i(k+1)=q_i|X(k)= x).
\end{aligned}
\end{equation}
Then, by explicitly evaluating each probability  $\textup{Pr}(X_i(k+1)=q_i|X(k)= x),$ using the dynamics \eqref{eq:dyn-ex-ih-update},  we obtain the desired result. 
\end{proof}

The result above generalizes that of \cite[Prop.~1]{ShuangPeterTransNNControl25} by including inhibitions in the transmission dynamics.

To further simplify the model, we introduce two other assumptions below. 
\begin{description}
	\item[\bf (A3)] (Transmission Independence at  Step $k$) For $k\geq 0 $, $\{ W_{ij}^k: i,j \in [n]\}$ are  independent and  for each $i, j \in [n]$, $W_{ij}^k$ is independent of  $\{ X_q(k): q \in [n],  q\neq j\}$. 
	\item[\bf (A4)] (State Independence among Neurons) {Up to some terminal step $T$, for each $k\leq T$, the states $\{X_i(k):i\in [n]\}$ are independent.}
\end{description}

Under (A1), (A3) and (A4),  taking the expectation on both sides of the equation \eqref{eq:dyn-ex-ih-update} yields
\begin{equation} \label{eq:simplified-prob}
\begin{aligned}
	 \mathbb{E} X_i(k+1)  
	& =\Big(1-\prod_{j\in E_i^{ \circ k}}\mathbb{E}(1- W_{ij}^k X_j(k)) \Big)  \\
	&\qquad \times  \prod_{j \in I_i^{ \circ k}}\mathbb{E}(1- W_{ij}^k X_j(k) )
\end{aligned} 
\end{equation}
for every $k \in\{0, \dots, T-1\}$.
 Let $w_{ij}^k\triangleq \operatorname{Pr}(W_{ij}^k=1 {| X_j(k)=1})$ denote the conditional probability of the successful transmission from node $j$ to node~$i$ at step $k$. Let  $p_i(k) \triangleq \textup{Pr}(X_i(k)=1)$. Then the equation \eqref{eq:simplified-prob} is equivalently represented by 
\begin{equation} \label{eq:prob-update-with-inhibition}
	\begin{aligned}
	 {p}_i(k+1) & = \Big(1-	 \prod_{j\in E_i^{ \circ k}} (1- w_{ij}^k {p}_j(k))\Big)\\
 	& \qquad \times \prod_{j\in I_i^{ \circ k}} (1- w_{ij}^k {p}_j(k)).
\end{aligned} 
\end{equation}
\begin{remark}
The assumption (A3) introduces the independence of the transmissions across different links and their independence of each link \(W_{ij}^k\) from the states of neurons other than its source neuron \(j\), and  (A3) is  more restrictive than (A2).  \textcolor{black}{(A4) can be satisfied with appropriate tree-like network structures and small time horizon $T$ that avoid state correlations. Such structures as well as the characterization of the error  incurred in terms of the second moment when (A4) does not hold will be presented in the journal extension.} 
\end{remark}
\subsection{Dynamics with State Transformation}

\textcolor{black}{The state transformation that simplifies the model in \cite{ShuangPeterTransNN22, ShuangPeterTransNNControl25,ShuangPeterTransNNCDC25}   does not apply directly to the model \eqref{eq:prob-update-with-inhibition} due to the existence of the inhibition probability}. 
To address this issue, we introduce  the probability of \textcolor{black}{no} inhibition at the previous step $k-1$ (from all neighboring neurons) at neuron $i$ below:
\begin{equation} \label{eq:inhibit_pi_def}
	\pi_i(k) \triangleq \prod_{j\in I_i^{\circ {k-1}} } (1- w_{ij}^{k-1} p_j(k-1))   \in [0,1].
\end{equation}
Then the probability update in \eqref{eq:prob-update-with-inhibition} is equivalently given by 
\begin{equation}\label{eq:prod-evo-update}
	\begin{aligned}
		p_i(k+1) = \Big(1- \prod_{j\in E_i^{\circ k}}(1- w_{ij}^kp_j(k)) \Big) \pi_i(k+ 1) 
	\end{aligned}
\end{equation}
which yields
\begin{equation}\label{eq:inhibit_dyn_transform}
	1- \frac{p_i(k+1)}{\pi_i(k+1)} = \textcolor{black}{\prod_{j\in E_i^{\circ k}}(1- w_{ij}^kp_j(k))}. 
\end{equation}
From \eqref{eq:prod-evo-update}, clearly ${p_i(k)}\leq {\pi_i(k)}$ holds for all $i \in  [n]$ and all $k\geq 1$.  
Define the following states  (of Shannon information) 
\begin{align}
	s_i(k)  &  \triangleq \begin{cases}
	     - \log\left(1-\frac{p_i(k)}{\pi_i(k)}\right) , & \text{ if } \pi_i(k) \in (0,1]  \\
          0 , & \text{ if } \pi_i(k) =0
	\end{cases}
   \label{eq:s-transform} \\
	o_i(k)  &  \triangleq - \log \pi_i(k)  \label{eq:o-transform}
\end{align}	
where $o_i(k), s_i(k)\in[0, +\infty]$ and  $\log$ function is defined by 
\begin{equation}
	\log (x) \triangleq 
	\begin{cases}
		\ln (x) ,&   x \in (0, 1];\\
		 -\infty, & x =0 .
	\end{cases}
\end{equation}
\begin{remark}
For neuron $i$, the state $o_i(k)$ is the Shannon information associated with the absence of inhibition from the neighboring neurons at the previous step $k-1$,  which can be intuitively understood as the inhibition level. In particular, the state $o_i(k) =0$ represents \textcolor{black}{no} active inhibition at $k-1$ and $o_i (k) =\infty$ represents  active inhibition at $k-1$.  
\textcolor{black}{For a neuron $i$, its state $s_i(k)$ is essentially the Shannon information associated with the absence of excitations from the excitatory neighboring neurons at the previous step $k-1$, that is
\[
s_i(k)= -\log \Big(1-	 \prod_{j\in E_i^{ \circ {k-1}}} (1- w_{ij}^{k-1} {p}_j(k-1))\Big).
\]}
%
The state $s_i(k) = 0$ if the neuron $i$ is resting with probability $1$ at time $k$ (i.e. $p_i(k)=0$), and $s_i(k) =+\infty$ if it fires with probability $1$ (i.e. $p_i(k)=1$). 
\end{remark}

\begin{remark}
The inverse mappings of the state transformations in  \eqref{eq:o-transform} and \eqref{eq:s-transform} are respectively given by
\[
\begin{aligned}
	& \pi_i(k) = e^{-o_i(k)} ~~\text{and}~~  p_i(k) =  e^{-o_i(k)} (1- e^{-s_i(k)}).
\end{aligned}
\]
\end{remark}

Taking logarithm  and negation on both sides of \eqref{eq:inhibit_dyn_transform} yields the following  representation of the evolution of $s_i(k)$ 
\begin{equation}%
\begin{aligned}
s_i(k+1) & = \sum_{j\in E_i^{\circ k}} \Psi(w_{ij}^k \pi_j(k), s_j(k))	
\end{aligned}
\end{equation}
where  
$
\Psi(w, x) \triangleq -\log(1-w+we^{-x})
$
is the Tuneable Log-Sigmoid (TLogSigmoid) activation function identified in \cite{ShuangPeterTransNN22}.
Replacing  $\pi_j(k)$ by $e^{-o_j(k)}$ from the relation \eqref{eq:o-transform} yields
\begin{equation}\label{eq:s-evo-update}
s_i(k+1)   = \sum_{j\in E_i^{\circ k}} \Psi(w_{ij}^k e^{-o_j(k)}, s_j(k)), \quad i \in [n]
\end{equation}
with initial condition	$s_i(0)  = - \log\left(1-{p_i(0)}\right) $ for all $i\in [n]$. 
 Furthermore, taking logarithm and negation on both sides of \eqref{eq:inhibit_pi_def} yields dynamics of state $o_i$ as follows:
\begin{equation}\label{eq:o-evo-update}
	o_i(k+1) = \sum_{j\in I_i^{\circ k} } \Psi(w_{ij}^k e^{-o_j(k)}, s_j(k)), \quad i \in [n]
\end{equation}
with initial condition $o_i(0)=0$ for all $i\in [n]$, if there is no inhibition before the initial step.

Equations \eqref{eq:s-evo-update} and \eqref{eq:o-evo-update} then completely characterize the evolution of the state $(s_i(k), o_i(k))$ over time step $k\geq 0$. 
The evolution can be computed as follows:
\begin{enumerate}
	\item Start with $o_i(0)=0$ for all $i \in [n]$, if there is no inhibition before step $k=0$.
\item Compute the state $s_i(0) = - \log\left(1-{p_i(0)}\right)$ based on the probability of excitation $p_i(0)$ for all node $i \in [n]$.
\item Iteratively compute the states $s_i$  and $o_i$ over time  based on the dynamics specified in \eqref{eq:s-evo-update} and \eqref{eq:o-evo-update}. 
\end{enumerate}

\begin{remark}[Convexity]
	We  highlight that if $w \in(0,1)$ and $x\in (0, \infty)$, the activation function $
\Psi(w, x) \triangleq -\log(1-w+we^{-x})
$ is strictly convex in $w $ and strictly concave in $x$, and is strictly monotonically increasing with respect to   $x$ and $w$ by evaluating the derivatives (see \cite[Section V]{ShuangPeterTransNN22}).
Let 
\[
g(z; v, s) \triangleq \Psi(v e^{-z}, s), \quad v \in (0,1), s \in (0,+\infty)
\]
Then 
$
\partial_z g =  \partial_w \Psi(ve^{-z}, s) v (-1)e^{-z} ,
$ and
\[
\begin{aligned}
	\partial_{zz}^2 g &  =   \partial_w \Psi(ve^{-z}, s) v e^{-z}+ \partial_{ww}^2 \Psi(ve^{-z}, s)     (v (-1)e^{-z}) ^2 \\
	& = ve^{-z}(\partial_w \Psi(ve^{-z}, s) + \partial_{ww}^2 \Psi(ve^{-z}, s)     v e^{-z}) >0
\end{aligned}
\]
since $\partial_w \Psi(ve^{-z}, s) >0$ and $\partial_{ww}^2 \Psi(ve^{-z}, s)>0$ (see \cite[Section V]{ShuangPeterTransNN22}).
This implies that the function $\Psi(v e^{-z}, s)$ is strictly convex in $z$ when $v \in (0,1), s \in (0,+\infty)$. 
\end{remark}

\section{Models with Neurotransmitter Populations}
To account for different realizations of the effective receptions of different neurotransmitters \cite{kandel2021principles} over the same link, we introduce the following generalized model: 
\begin{equation}\label{eq:population-model}
\begin{aligned}
	 X_i(k+1) =   & \Big(1-\prod_{j\in E_i^{\circ k}} \prod_{\ell =1}^{a_{ij}^k} (1- W_{ij (\ell) }^{k} X_j(k)) \Big) \\
	& \quad \times  \prod_{j \in I_i^{\circ k}} \prod_{\ell =1}^{a_{ij}^k} (1- W_{ij (\ell)}^k X_j(k) )\end{aligned}
\end{equation}
where $a_{ij}^k$ denotes the number of neurotransmitters sent from neuron $j$ to neuron $i$ at step $k$, and $W_{ij(\ell)}^k$ is a binary variable that represents the successful reception of the $\ell^{th}$ neurotransmitter at step $k$ from neuron $j$ to neuron $i$ when taking $1$, otherwise $0$.  
In this way, the successful reception of a neurotransmitter represented by a binary random variable $ W_{ij (\ell) }^{k}$ is realized at each transmission $\ell$ at step $k$ from neuron $j$ to neuron $i$.   

We introduce the following assumptions regarding independence of neurotransmissions to identify an approximate model for the probability of excitation. 
\begin{description}
\item[\bf(A5)]  (Memoryless Neurotransmission)  For any $k> 0$,	 $W^k\triangleq [W_{ij(\ell)}^k]\in \BR^{n\times n \times a_{ij}^k}$ is independent of $\{W^t: 0\leq t < k\}$ and $\{X{(t)}: 0\leq t < k \}$. 
\item[\bf (A6)]  (Neurotransmission Conditional Independence) For any $k\geq 0$, the  binary random variables $\{W_{ij (\ell)}^k:i,j \in [n], \ell \in  [a_{ij}^k]\}$ representing the transmissions are jointly conditionally independent given the current states $X(k)$. 
\end{description}
Following similar steps of the previous section, under assumptions (A5) and (A6), we obtain
\begin{equation}
\begin{aligned}
	 \mathbb{E}&(X_i(k+1)| (X(t))_{t\in [k]}) 
	 =  \mathbb{E}(X_i(k+1)| X(k))\\
	& =\Big(1-\prod_{j\in E_i^{ \circ k}}\prod_{\ell =1}^{a_{ij}^k}\mathbb{E}(1- W_{ij(\ell)}^k X_j(k)| X(k)) \Big)  \\
	&\qquad \times  \prod_{j \in I_i^{ \circ k}}\prod_{\ell =1}^{a_{ij}^k}\mathbb{E}(1- W_{ij(\ell)}^k X_j(k) | X(k) ) .
\end{aligned}
\end{equation} 

\begin{proposition}\label{prop:cond-prob-infection-population}
Assume (A5) and (A6) hold. 
Given a state configuration ${x} \in \{0,1\}^n$ at step $k$, the transition probability to a state configuration $q\in\{0,1\}^n$ is  given by 
\begin{equation}
\begin{aligned}
	\textup{Pr}(&X(k+1)=q|X(k)=x)\\ 
	& = \prod_{i=1}^n \Big(q_i\rho_i(k+1) + (1-q_i) (1-\rho_i(k+1)) \Big)
\end{aligned}
\end{equation}
where
 \begin{equation}\label{eq:prob-configuration-population}
 \begin{aligned}
\rho_i(k+1) & = \Big(1-	 \prod_{j\in E_i^{ \circ k}} (1- w_{ij}^k {x}_j)^{a_{ij}^k}\Big)\\
 	& \quad \times \prod_{j\in I_i^{ \circ k}} (1- w_{ij}^k {x}_j)^{a_{ij}^k}, \quad i\in [n].
 \end{aligned}
\end{equation}
\end{proposition}
\textcolor{black}{The proof is similar to that of Prop.~\ref{prop:cond-prob-infection} and hence omitted.}
\begin{remark}
Compared to 	Proposition~\ref{prop:cond-prob-infection}, the difference lies in the representation of the probability $\rho_i(k+1) $ in \eqref{eq:prob-configuration-population} which now involves  the number of neurotransmitters $a_{ij}^k$ with $ i, j \in [n]$ and $k\geq 0$. 

Such results that evaluate transition probabilities are  needed for computing optimal control solutions under the framework of Markov decision processes   \cite{ShuangPeterTransNNControl25, ShuangPeterTransNNCDC25}. 
\end{remark}

We introduce the following assumptions to further simplify the representation.
\begin{description}
	\item[\bf (A7)] (Neurotransmission Independence at  Step $k$)  At step $k\geq 0 $, $\{ W_{ij(\ell)}^k: i,j \in [n], \ell \in [a_{ij}^k]\}$ are  independent,  and  for each $i, j \in [n]$ and each $\ell \in [a_{ij}^k]$, $W_{ij(\ell)}^k$ is independent of  $\{ X_q(k): q\in [n],  q\neq j\}$. 
	\item[\bf (A8)] (State Independence among Neurons) {Up to some terminal step $T$, for each $k\leq T$, the underlying binary random variables $\{X_i(k):i\in [n] \}$ are independent.}
\end{description}

Under (A5), (A7) and (A8), taking the expectation on both sides of the equation \eqref{eq:population-model} above yields
\begin{equation} 
\begin{aligned}
 \mathbb{E}	 X_i(k+1) =   & ~ \mathbb{E}	\Big(1- \prod_{j\in E_i^{ \circ k}} \prod_{\ell =1}^{a_{ij}^k} (1- W_{ij (\ell) }^{k} X_j(k)) \Big) \\
	& \times \textcolor{black}{\mathbb{E}} \prod_{j \in I_i^{ \circ k}} \prod_{\ell =1}^{a_{ij}^k} (1- W_{ij (\ell)}^k X_j(k)) ).
\end{aligned}
\end{equation}
Denote  $p_i(k) \triangleq \textup{Pr}(X_i(k)=1)$. Let $w_{ij}^k\triangleq \operatorname{Pr}(W_{ij (\ell)}^k=1 {| X_j(k)=1})$ denote the conditional probability of the successful reception of each neurotransmitter  from node $j$ to node~$i$ at step $k$. 
Then the equation above is equivalent to
\[
\begin{aligned}
	 p_i(k+1) & =\Big(1 -\prod_{j\in E_i^{ \circ k}}(1- w_{ij}^k p_j(k))^{a_{ij}^k}\Big)\\
	&  \quad \times     \prod_{j\in I_i^{ \circ k}}(1- w_{ij}^k p_j(k)) )^{a_{ij}^k} .
\end{aligned} 
\]
With a slight abuse of notation, define 
\begin{align}
\label{eq:inhibit_pi_def_population}
	\pi_i(k) & \triangleq \prod_{j\in I_i^{\circ {k-1}} } (1- w_{ij}^{k-1} p_j(k-1))^{a_{ij}^{k-1}}   \in [0,1] \\
	s_i(k)  & \triangleq  \begin{cases}
	     - \log\left(1-\frac{p_i(k)}{\pi_i(k)}\right) , & \text{ if } \pi_i(k) \in (0,1]  \\
        0, & \text{ if } \pi_i(k) =0
	\end{cases} \label{eq:s-transform-pop} \\
	o_i(k)  &  \triangleq - \log \pi_i(k)  \in [0,+\infty]. \label{eq:o-transform-pop}
\end{align}	
with $o_i(k), s_i(k)\in[0, +\infty]$.

Following the same analysis  in the previous section, we obtain the following dynamics
\begin{align}
s_i(k+1)  & = \sum_{j\in E_i^{\circ k}} a_{ij}^k \Psi(w_{ij}^k e^{-o_j(k)}, s_j(k)), \quad i \in [n] \label{eq:s-evo-update-pop}\\
	o_i(k+1) & = \sum_{j\in I_i^{\circ k} } a_{ij}^k \Psi(w_{ij}^k e^{-o_j(k)}, s_j(k)), \quad i \in [n]\label{eq:o-evo-update-pop}
\end{align}
which explicitly include the numbers of neurotransmitters $\{a_{ij}^k: i, j \in [n], k\geq 0\}$ into the evolution dynamics.

The initial conditions can be given by $o_i(0)=0$ (if no inhibition exists before the starting time), and $s_i(0)  = - \log\left(1-{p_i(0)}\right) $, for all $i\in [n]$.

\begin{remark}
	The   equation pair \eqref{eq:s-evo-update} and \eqref{eq:o-evo-update} can now be viewed as special cases of the equation pair \eqref{eq:s-evo-update-pop} and \eqref{eq:o-evo-update-pop} by setting $a_{ij}^k=1$ for all $i, j \in [n]$ and $k\geq 0$.
\end{remark}

\begin{remark}[Biological Interpretation]
    An important feature of the model above is that the connection weights $a_{ij}^k$ and $w_{ij}^k$ are non-negative and \textcolor{black}{admit direct biological} interpretations \textcolor{black}{(i.e. the number of neurotransmitters and the conditional probability of the successful reception of each neurotransmitter, respectively)}, \textcolor{black}{in contrast to conventional} neural networks that allow negative weights \cite{block1962perceptron,rumelhart1986learning}.  
    
Such dynamics in \eqref{eq:s-evo-update-pop} and \eqref{eq:o-evo-update-pop}   are related to graph neural networks \cite{scarselli2009graph}  where each node has a continuous  state (or nodal feature) of dimension two. The state $s_i$ (resp. $o_i$) summarizes the incoming influence of excitatory (resp. inhibitory) connections.
\end{remark}
\begin{remark}
An offset  in the dynamics \eqref{eq:s-evo-update-pop} and \eqref{eq:o-evo-update-pop} can be created  by introducing one node $n_0 \in [n]$ with no incoming links and with the outgoing transmission probability $w_{in_0}^k <1$,  for $i\in [n]$ and $k\geq 0$, and setting its state $s_{n^0}(k)$ to $\infty$. 

\end{remark}

Let $\odot$ denote Hadamard product and introduce 
$$A^k=[a_{ij}^k] ,\quad  \Omega^k=[w_{ij}^k], \quad M^{k} =  A^k \odot \Omega^k$$ 
which are  $n\times n$ matrices.
Then we  have the following upper bounds for the states of \eqref{eq:s-evo-update-pop} and \eqref{eq:o-evo-update-pop}.
\begin{proposition}[Upper Bound] \label{prop: upper-bound-s}
Let the initial states be given by ${s}(0) = [{s}_1(0), \cdots, {s}_n(0)]^\TRANS$ with ${s}_i(0)  = - \log\left(1-{p_i(0)}\right) $ and ${o}(0)= [{o}_1(0), \cdots, {o}_n(0)]^\TRANS$.
Then the states of \eqref{eq:s-evo-update-pop} and \eqref{eq:o-evo-update-pop} for  $k\geq 1$ satisfy that for all $i \in [n]$, 
\begin{align}
{s}_i(k)  & \leq  [\mathcal{T}_E(k,0)  {s}(0)]_i \\ {o}_i(k) &\leq   [ (B_I^{k-1} \odot M^{k-1} ) \mathcal{T}_E(k-1,0)   {s}(0)]_i
\end{align}
with $\mathcal{T}_E(k,0) \triangleq (B_E^{k-1}\odot M^{k-1}) \cdots (B_E^0 \odot M^0)$.
\end{proposition}
See Appendix for the proof. 

\section{Limit Model with Infinite Neurotransmitters}
 \subsection{Limit Model via Poisson  Approximation}
Since the number of released neurotransmitters could be very large (see e.g. \cite[Part III, Chp. 11]{kandel2021principles}), the number of receptors at the post-synaptic neurons are relatively moderate and are assumed not to change over a short period of time,  the probability of a successful transmission $w_{ij}$ of each neurotransmitter from node $j$ to node $i$ may decrease with respect to the number of neurotransmitters $a_{ij}$. 
Hence  we introduce the following assumption.
\begin{description}
	\item[\bf (A9)]  
	The probability of transmission $w_{ij}^k$ depends on the number of transmissions $a_{ij}^k$ as follows:
\begin{equation*}
	w_{ij}^k = \frac{\lambda_{ij}^k}{a_{ij}^k},\quad \text{with $\lambda_{ij}^k$ fixed } \forall k\geq 0, ~ \forall i,j \in[n].  
\end{equation*}

\end{description}

\begin{remark}[Poisson Approximation] Consider $n$ independent Bernoulli random variables $Z_1,\cdots, Z_n$, each of which takes $1$ with probability  $\frac{\lambda}{n}$. Then the sum $S = \sum_{i=1}^n Z_i$ follows a binomial distribution and  hence can be approximated by a Poisson distribution with rate $\lambda$.
Then 
$
\text{Pr}\left(\prod_{i=1}^n (1-Z_i) =1\right) =  \text{Pr}(S=0 ) \approx e^{- \lambda}. 
 $
\end{remark}

Following the Poisson approximation of binomial distributions,  if (A9) holds, 
\[
\text{Pr}(W_{ij(q)}^k X_j(k) =1) = w_{ij}^k p_j(k) = \frac{\lambda_{ij}^kp_j(k)
}{a_{ij}^k}. 
\]
Applying the Poisson approximation yields
\begin{equation} \label{eq:poisson-approx}
	\text{Pr}\left(\prod_{q=1}^{a_{ij}^k} \Big(1- W_{ij(q)}^k X_j(k)\Big)=1\right) \approx e^{{-\lambda_{ij}^k} p_j(k)}.
\end{equation}
Under (A5), (A7), (A8) and (A9), with the Poisson approximation, the expected state then satisfies
\begin{equation}\label{eq:agent-based-population-poisson}
\begin{aligned}
\mathbb{E}X_i(k+1)  
&= 	\Big(1- \prod_{j\in {E}_i^{\circ k}} \mathbb{E}\prod_{q=1}^{a_{ij}^k} \big(1- W_{ij(q)}^k X_j(k)\big)\Big)\\
	& \qquad \times  \prod_{j \in I_i^{ \circ k}} \mathbb{E}  \prod_{\ell =1}^{a_{ij}^k} (1- W_{ij (\ell)}^k X_j(k)) )\\
& ~ \approx ~ \Big(1- \prod_{j\in {E}_i^{\circ k}} e^{{-\lambda_{ij}^k} p_j(k)}\Big)\times \prod_{j\in {I}_i^{\circ k}} e^{{-\lambda_{ij}^k} p_j(k)}, 
\end{aligned}
\end{equation}
that is 
\begin{equation}\label{eq:dyn_prob_poisson}
\begin{aligned}
p_i(k+1)  
& ~ \approx ~ \Big(1- \prod_{j\in {E}_i^{\circ k}} e^{{-\lambda_{ij}^k} p_j(k)}\Big)  \prod_{j\in {I}_i^{\circ k}} e^{{-\lambda_{ij}^k} p_j(k)}
\end{aligned}
\end{equation}
for all $i \in [n]$,
where $\lambda_{ij}^k = w_{ij}^ka_{ij}^k$ is the  rate of the Poisson distribution at time $k$ for the synaptic connection from neuron $j$ to neuron~$i$.  
Heuristically, this approximation works well when $a_{ij}^k$ is large, $ w_{ij}^k$ is small, and $ \lambda_{ij}^k = a_{ij}^k w_{ij}^k$ is moderate. 
To simplify the representation of the dynamics \eqref{eq:dyn_prob_poisson}, let  
\begin{equation} \label{eq:pi-bar-pop}
	\bar{\pi}_i(k) \triangleq \prod_{j\in {I}_i^{\circ {k-1}}} e^{{-\lambda_{ij}^{k-1}} p_j(k-1)}  \in (0,1]
\end{equation} 
 represent the probability of \textcolor{black}{no  inhibition} at node $i$ from its neighbouring neurons in the previous step. 
Let 
\begin{align}
	\bar{s}_i(k)  &  \triangleq - \log\left(1-\frac{p_i(k)}{\bar{\pi}_i(k)}\right) \in[0, +\infty] \label{eq:s-transform-pop} \\
	\bar{o}_i(k)  &  \triangleq - \log \bar{\pi}_i(k)  \in [0,+\infty). \label{eq:o-transform-pop}
\end{align}	
From \eqref{eq:dyn_prob_poisson}, we obtain the dynamics for $(\bar{s}, \bar{o})$, given by
\begin{align}
	\bar{s}_{i}(k+1) 
	& =   \sum_{j\in {E}_i^{\circ k}}  \lambda_{ij}^k e^{-\bar{o}_j(k)} (1-e^{-\bar{s}_j(k)}) \label{eq:sbar}\\
		\bar{o}_i(k+1) & =  \sum_{j \in I_i^{\circ k}}  \lambda_{ij}^k e^{-\bar{o}_j(k)} (1-e^{-\bar{s}_j(k)}).
\label{eq:obar}
\end{align}
 The initial conditions are given by $\bar{o}_i(0)=0$ (representing no inhibitions before the first step) and $\bar{s}_i(0)  = - \log\left(1-{p_i(0)}\right) $, for all $i\in [n]$.
\begin{remark}
    \textcolor{black}{We note that $\bar{\pi}_i(k)$ cannot be zero following its definition in \eqref{eq:pi-bar-pop}, since $\lambda_{ij}^k$ is assumed to be finite, and  $p_{j}(k-1) \in [0,1]$. Therefore, we excluded $+\infty$ in \eqref{eq:o-transform-pop}.  }  
\end{remark}
\begin{proposition}\label{prop:limit-os-bar}
    
 Assume (A5), (A7), (A8) and (A9) hold.  Then the limit model for the probability of excitation for the dynamics \eqref{eq:population-model}, when $a_{ij}^k \to \infty$ for all $i,j\in [n]$ and $k\in \{0,\cdots, T-1\}$,  is  given by \eqref{eq:sbar} and \eqref{eq:obar} with $$\textup{Pr}(X_i(k)=1) = e^{-\bar{o}_i(k)} (1-e^{-\bar{s}_i(k)}),$$
 where $k\in \{0,\cdots, T\}$.
\end{proposition}
\begin{proof}
Under (A5), (A7), (A8) and (A9),   
\[
\begin{aligned}
& \lim_{a_{ij}^k\to \infty} \mathbb{E}\prod_{q=1}^{a_{ij}^k} \big(1- W_{ij(q)}^k X_j(k)\big) \\
&=\lim_{a_{ij}^k\to \infty} \text{Pr} \Bigg(\prod_{q=1}^{a_{ij}^k} \Big(1- W_{ij(q)}^k X_j(k)\Big)=1\Bigg) \\
 &= \lim_{a_{ij}^k\to \infty} \left( 1- {w_{ij}^k} p_j(k)\right)^{a_{ij}^k}\\ & = \lim_{a_{ij}^k\to \infty} \left( 1- \frac{\lambda_{ij}^k}{a_{ij}^k} p_j(k)\right)^{a_{ij}^k} 
= e^{{-\lambda_{ij}^k} p_j(k)}.
\end{aligned}
\]
that is, the  equality in \eqref{eq:poisson-approx} is  exact when the number of neurotransmitters goes to infinity on each link.
The rest of the proof follows by replacing the approximate equalities in \eqref{eq:agent-based-population-poisson} and \eqref{eq:dyn_prob_poisson} by  exact equalities.
%
\end{proof}

To facilitate further analysis, we introduce the element-wise activation function $\phi: \BR^n \times \BR^n \to \BR^n$ defined  below
$$
\phi (s, o) \triangleq [\sigma(s_1, o_1), \cdots, \sigma(s_n, o_n)]^\TRANS \in \BR^n, \quad \forall s, o \in \BR^n
$$ with 
$
\sigma (s_i, o_i) \triangleq  e^{-o_i} (1-e^{-s_i}) 
$ for any $s_i, o_i\in \BR$.
Let   $\bar{s}(k ) = [\bar{s}_1(k),\cdots, \bar{s}_n(k)]^\TRANS$ and  $\bar{o}(k ) = [\bar{o}_1(k),\cdots, \bar{o}_n(k)]^\TRANS$. 
Then the dynamics in \eqref{eq:sbar} and \eqref{eq:obar} can be  represented in a compact form below
\begin{equation}\label{eq:compact-barso}
    \MATRIX{\bar{s}(k+1)\\\bar{o}(k+1)}= \MATRIX{B_E^k \odot \Lambda^k \\  B_I^k \odot\Lambda^k } \phi(\bar{s}(k), \bar{o}(k))
\end{equation}
where $\Lambda^k \triangleq [\lambda_{ij}^k] \in \BR^{n\times n}$ and $\odot$ denotes Hadamard product. 
 We note that the diagonal elements of binary-valued adjacency matrices $B_E^k$ and $B_I^k$ could be non-zero due to the existence of self-loops (representing autapses).

Let $p(k)\triangleq [p_1(k),\cdots, p_n(k)]^\TRANS$. For any vector $v\in \BR^n$, its exponential function is defined by $e^{v} = [e^{v_1},\cdots, e^{v_n}]^\TRANS$. 
\begin{proposition}
 Assume (A5), (A7), (A8) and (A9) hold.  Then the limit model for the probability of excitation for the dynamics \eqref{eq:population-model}, when $a_{ij}^k \to \infty$ for all $i,j\in [n]$ and  $k\in \{0,\cdots, T-1\}$,  is  given by     \begin{equation} \label{eq:prob-vec-evo}
    p(k+1) = (1- e^{-(B_E^k \odot \Lambda^k) p(k)} )\odot e^{-(B_I^k \odot \Lambda^k) p(k)}
\end{equation}  with  $\textup{Pr}(X_i(k)=1) =p_i(k)$ for all $i\in [n]$, where  $k\in \{0,\cdots, T-1\}$.
\end{proposition}

\begin{proof}
    We follow the same proof steps in Prop.~\ref{prop:limit-os-bar} to establish an exact equality in \eqref{eq:dyn_prob_poisson}, for which the  compact representation is equivalently given by \eqref{eq:prob-vec-evo}.
\end{proof}

 A trivial equilibrium  for \eqref{eq:prob-vec-evo} is $p^* = [0,\cdots, 0]^\TRANS \in \BR^{n}$. 

\subsection{Contraction and Stability Properties}

\begin{proposition}[Contraction] \label{prop:contraction}
Let $p = 1 \text{ or } \infty$.
     If 
\begin{equation} \label{eq:contraction}
\left\|\MATRIX{B_E^k \odot \Lambda^k \\  B_I^k \odot\Lambda^k }  \right\|_p < 1, \quad \forall k \geq 0
\end{equation}
 holds, 
     the system in \eqref{eq:sbar} and \eqref{eq:obar} is contracting, that is, 
\begin{equation} \label{eq:contraction-per-step}
    \left\| \MATRIX{\bar{s}(k+1) \\ \bar{o}(k+1)} -  \MATRIX{\bar{s}^*(k+1) \\ \bar{o}^*(k+1)} \right \|_p <  \left\| \MATRIX{\bar{s}(k) \\ \bar{o}(k)} -  \MATRIX{\bar{s}^*(k) \\ \bar{o}^*(k)} \right \|_p
\end{equation}
and
\begin{equation} \label{eq:contraction-from-start}
    \left\| \MATRIX{\bar{s}(k) \\ \bar{o}(k)} -  \MATRIX{\bar{s}^*(k) \\ \bar{o}^*(k)} \right \|_p <  \left\| \MATRIX{\bar{s}(0) \\ \bar{o}(0)} -  \MATRIX{\bar{s}^*(0) \\ \bar{o}^*(0)} \right \|_p
\end{equation}
where $\MATRIX{\bar{s}(k) \\ \bar{o}(k)} $ (resp. $\MATRIX{\bar{s}^*(k) \\ \bar{o}^*(k)}$ ) denotes the state at step $k$ for the  trajectory with the initial value $\MATRIX{\bar{s}(0) \\ \bar{o}(0)} $ (resp. $\MATRIX{\bar{s}^*(0) \\ \bar{o}^*(0)}$).
\end{proposition}
See Appendix for the proof. 

\begin{proposition}[Upper Bound] \label{prop: upper-bound-bars}
Let the initial states be given by $\bar{s}(0) = [\bar{s}_1(0), \cdots, \bar{s}_n(0)]^\TRANS$ with $\bar{s}_i(0)  = - \log\left(1-{p_i(0)}\right) $ and $\bar{o}(0)= [\bar{o}_1(0), \cdots, \bar{o}_n(0)]^\TRANS$. Then the states of the system \eqref{eq:sbar} and \eqref{eq:obar} for any step  $k\geq 1$ satisfy that for all $i \in [n]$, 
\begin{align}
\bar{s}_i(k)  & \leq  [\Gamma_E(k,0)  \bar{s}(0)]_i \\ \bar{o}_i(k) &\leq   [ (B_I^{k-1} \odot \Lambda^{k-1}) \Gamma_E(k-1,0)   \bar{s}(0)]_i
\end{align}
with $\Gamma_E(k,0) \triangleq (B_E^{k-1} \odot \Lambda^{k-1}) \cdots (B_E^0 \odot \Lambda^0)$. 
\end{proposition}
\begin{proof}
Using the property that $1-e^{-x} \leq x$ for $x\geq 0$ and then following the same proof steps of Prop.~\ref{prop: upper-bound-s} to build the linear dynamical system that provides the upper bounds for the states, we can easily obtain the desired results.
\end{proof}
\begin{proposition}[Stability] \label{prop:stability-bars}
    Assume $B_E^k = B_E$,  $B_I^k = B_I$ and $\Lambda^k = \Lambda$ are invariant with respect to the step $k\geq 0$. Then the system \eqref{eq:sbar} and \eqref{eq:obar}  is asymptotically and exponentially stable with respect to the step $k$ at the origin if 
    \[
    \max_{i \in [n]}|\lambda_i(B_E \odot \Lambda)| <1
    \]
    with $\{\lambda_i(B_E \odot \Lambda), i \in [n]\}$ as all the eigenvalues of $B_E\odot \Lambda$. 
\end{proposition}
\begin{proof}
    Prop.~\ref{prop: upper-bound-bars} together with the condition for the stability of   discrete-time linear systems implies the desired result.
\end{proof}
\begin{remark}
   \textcolor{black}{Interestingly,  the contraction property in Prop.~\ref{prop:contraction} involves both the excitatory and inhibitory networks, whereas the stability property in Prop.~\ref{prop:stability-bars} depends only on the excitatory network.
   For the inhibitory connections to be effective, the corresponding neurons that initiate the inhibitions need to be active, whose activation in the present model is governed by the excitatory connections.  Hence the excitatory connections play a primary role in determining the stability properties of the system in Prop.~\ref{prop:stability-bars}. }
\end{remark}
\section{Conclusion}
We have generalized  the TransNN model to include inhibitions and  shown that under technical assumptions, the probability of neuronal excitations in TransNNs with both inhibitory and excitatory connections  can be equivalently represented by neural networks where each neuron has a two-dimensional continuous state  and each link has the TLogSigmoid activation function. Moreover,   neurotransmitter populations \textcolor{black}{were} considered in an extended model, and  Poisson approximation \textcolor{black}{was} applied to establish limit models when the number of neurotransmitters at each link \textcolor{black}{goes to infinity}.  Sufficient conditions for stability and contraction properties of the limit network model have been established. 

Future work should investigate non-trivial equilibria,  the integration of the neuronal dynamics for action potentials and spiking sequences, 
the control of such networks, and the game-theoretic analysis of such dynamics with costs  (such as energy and resources for firing).  
\bibliographystyle{IEEEtran}
\bibliography{mybib}
\appendix[Proof of Proposition~\ref{prop: upper-bound-s}]
\begin{proof} By the concavity of $
\Psi(w, x) \triangleq -\log(1-w+we^{-x})
$ in $x$ (see \cite[Sec.~V]{ShuangPeterTransNN22}), we have for any $z ,z^* \in [-\infty, +\infty]$,
$$\Psi(w,z)\leq  \Psi(w,z^*) + \partial_x\Psi(w,z^*)(z-z^*), \quad w \in [0,1] .
$$ In particular, taking $z^*=0$ yields $\partial_x\Psi(w,0)=w$ and hence $\Psi(w,z) \leq w z$. 
   Applying this inequality to  \eqref{eq:s-evo-update-pop}  yields
    \begin{align*}
	{s}_{i}(k+1)
	& \leq   \sum_{j\in {E}_i^{\circ k}}  a_{ij}^k w_{ij}^k e^{-o_j(k)} {s}_j(k) \leq  \sum_{j\in {E}_i^{\circ k}}  a_{ij}^k w_{ij}^k {s}_j(k). 
\end{align*}
for $o_j(k)\in[0,+\infty]$. 
 That is, the state ${s}$ is element-wise upper bounded  by the state of the  discrete-time linear  system
\[
 z(k+1) = (B_E^k \odot M^k)   z(k), \quad z(0) = {s}(0),~~ z(k) \in \BR^n,
\]
the solution of which is given by  $z(k ) = \mathcal{T}_E(k,0){s}(0)$. 
Thus, for all $k\geq 0$,
$
{s}_i(k) \leq z_i(k) =  [\mathcal{T}_E(k,0)  {s}(0)]_i. 
$
 In addition, from the dynamics \eqref{eq:o-evo-update-pop}  of the state ${o}$,  we obtain similarly
    \begin{align}
	{o}_{i}(k+1)
	\leq  \sum_{j\in {I}_i^{\circ k}}  a_{ij}^k w_{ij}^k {s}_j(k), \quad \forall i \in [n]. 
\end{align}
Therefore,
\[
\begin{aligned}
{o}_i(k+1) \leq [ (B_I^{k} \odot M^{k})   {s}(k)]_i \leq [ (B_I^{k} \odot M^{k})   z(k)]_i   .
\end{aligned}
\] 
Replacing $z(k)$ by its solution $\mathcal{T}_E(k,0){s}(0)$ and shifting the time index from $k+1$ to $k$ yield the desired result.
\end{proof}
\appendix[Proof of Proposition~\ref{prop:contraction}]
\begin{proof}
    We follow the  contraction analysis  for dynamical systems (see  e.g. \cite{Bullo_contraction,lohmiller1998contraction}). 
    The gradients of $\phi(s,o)$ satisfy 
\[
\begin{aligned}
\frac{\partial \phi}{\partial s } (s, o) &  = \text{diag}(e^{-o_1}e^{-s_1}, \cdots, e^{-o_n}e^{-s_n}) \\
   \frac{\partial \phi}{\partial o}(s, o)  & = \text{diag}(-e^{-o_1}(1-e^{-s_1}), \cdots, -e^{-o_n}(1-e^{-s_n})) \\
\end{aligned}
\]
  Then the following hold
\begin{equation}\label{eq:phi-derivative}
\left\| \MATRIX{\frac{\partial \phi}{\partial s}(s, o) &  \frac{\partial \phi}{\partial o}(s, o)}\right\|_p  \leq \max_{i\in  [n]} e^{-o_i}\leq 1
\end{equation}
for all $s, o \in [0,+\infty]^n$ and $p \in \{ 1,  \infty\}$.
Let \[
F^k(s,o)\triangleq  \MATRIX{B_E^k \odot \Lambda^k \\  B_I^k \odot\Lambda^k } \phi(s, o), \quad s, o \in \BR^n .
\]
The Jacobian of the system \eqref{eq:compact-barso} satisfies 
\[
\begin{aligned}
    J^k (s,o)  & \triangleq \MATRIX{\frac{\partial F^k}{\partial s} (s,o) &  \frac{\partial F^k}{\partial o} (s,o)} \\
    &  = \MATRIX{B_E^k \odot \Lambda^k \\  B_I^k \odot\Lambda^k }  \MATRIX{\frac{\partial \phi}{\partial s} (s, o) &  \frac{\partial \phi}{\partial o} (s, o)}. 
\end{aligned}
\]
Then the submultiplicativity of the induced norms ($\|\cdot\|_1$ and $\|\cdot\|_\infty$), together with \eqref{eq:contraction} and \eqref{eq:phi-derivative}, implies 
\begin{equation} \label{eq:jacobian-norm}
    \sup_{s, o \in [0,+\infty]^n}\|J^k(s,o)\|_p  \leq \left\|\MATRIX{B_E^k \odot \Lambda^k \\  B_I^k \odot\Lambda^k }  \right\|_p < 1.
\end{equation}
    We introduce the following notational simplification:
    $$
    \begin{aligned}
     & y(k)\triangleq [\bar{s}(k),\bar{o}(k)]^\TRANS,  {\quad  } y^*(k)\triangleq [\bar{s}^*(k),\bar{o}^*(k)]^\TRANS \\
     &   \delta y(k)\triangleq y(k) - y^*(k),  \quad f(y(k)) \triangleq F^k(\bar{s}(k), \bar{o}(k)).
    \end{aligned}
    $$
Then $\delta y(k+1) = f(y(k)) -  f(y^*(k))$. 
Define $$g(\tau) \triangleq f(y^*(k) + \tau(y(k) -y^*(k)))$$ for $\tau \in [0,1]$. Then 
$
\delta y(k+1) = f(y(k)) -  f(y^*(k))  =  g(1) - g(0) = \int_0^1 g^\prime(\tau) d \tau.
$
Furthermore, since 
$$
\begin{aligned}
    g^\prime (\tau) &  = \frac{\partial}{\partial y} f(y^*(k)+ \tau (y(k)-y^*(k))) (y(k) - y^*(k)) \\
    &  =\frac{\partial}{\partial y} f(y^*(k)+\tau \delta y(k)) \delta y(k),
\end{aligned}
$$
we obtain
$
    \delta y(k+1)  
     = \int_0^1 \frac{\partial}{\partial y} f(y^*(k)+\tau \delta y(k))d\tau ~ \delta y(k).
$
Let $\bar{s}_\tau(k)\triangleq \bar{s}^*(k) + \tau (\bar{s}(k)- \bar{s}^*(k))$ and $\bar{o}_\tau(k)$ defined similarly. That is, $y^*(k)+\tau \delta y(k) = [\bar{s}_\tau(k),\bar{o}_\tau(k)]^\TRANS $. Then
\[
\begin{aligned}
& \frac{\partial}{\partial y} f(y^*(k)+\tau \delta y(k)) = J^k(\bar{s}_\tau(k),\bar{o}_\tau(k)) \\ 
& =   \MATRIX{B_E^k \odot \Lambda^k \\  B_I^k \odot\Lambda^k } \MATRIX{\frac{\partial \phi}{\partial s}(\bar{s}_\tau(k), \bar{o}_\tau(k)) &   \frac{\partial \phi}{\partial o}(\bar{s}_\tau(k), \bar{o}_\tau(k))}. 
\end{aligned}
\]
Hence 
\begin{equation*}
\begin{aligned}
   \|\delta  & y(k+1)\|_p  
     \leq \int_0^1 \left\|\frac{\partial}{\partial y} f(y^*(k)+\tau \delta y(k))\right\|_pd\tau ~ \|\delta y(k)\|_p \\
     & = \int_0^1 \left\|J^k(\bar{s}_\tau(k),\bar{o}_\tau(k)))\right\|_pd\tau ~ \|\delta y(k)\|_p  < \|\delta y(k)\|_p,
\end{aligned}
\end{equation*}
where the last inequality is due to \eqref{eq:jacobian-norm}. Therefore \eqref{eq:contraction-per-step} holds. The inequality in \eqref{eq:contraction-from-start} follows  by iteratively applying  the inequality \eqref{eq:contraction-per-step}. 
\end{proof}
\end{document}